\documentclass{amsart}%
\usepackage{amsfonts}
\usepackage{amsmath}
\usepackage{amssymb}
\usepackage{color}
\usepackage[polish,english]{babel}
\usepackage[utf8]{inputenc}
\usepackage[T1]{fontenc}
\usepackage{graphicx}%
\providecommand{\U}[1]{\protect\rule{.1in}{.1in}}
\theoremstyle{plain}
\newtheorem{theorem}{Theorem}[section]
\newtheorem{lemma}[theorem]{Lemma}
\newtheorem{proposition}[theorem]{Proposition}
\theoremstyle{definition}

\theoremstyle{remark}
\newtheorem{remark}[theorem]{Remark}
\numberwithin{equation}{section}

\begin{document}
\title{The structure of rationally factorized Lax type flows and their analytical integrability}
\author{Myroslava \ I. Vovk, \ Petro Pukach, Oksana Ye. Hentosh$^{1)}$ and Yarema A.
Prykarpatsky$^{2)}$}
\address{$^{1)}$The Institute for Applied Problems of Mechanics and Mathematics at the
NAS, Lviv, 79060 Ukraine\\
$^{2)}$The Department of Applied Mathematics at the University of Agriculture
in Krakow, 30059, Poland}
\email{ohen@ukr.net, yarpry@gmail.com}

\begin{abstract}
In the article we construct a wide class of differential-functional dynamical
systems, whose rich algebraic structure makes their integrability analytically
effective. In particular, there is analyzed in detail the operator Lax type
equations for factorized seed elements, there is proved an important theorem
about their operator factorization and the related analytical solution scheme
to the corresponding nonlinear differential-functional dynamical systems.

\end{abstract}
\maketitle

\section{\label{Sec_1}The basic associative algebra case}

There is considered an associative functional algebra ${\mathcal{A}}\subset
C^{\infty}(\mathbb{S}^{1};\mathbb{C}),$ admitting the automorphism
\begin{equation}
T\circ a(x):=a(x+\delta i)\ \label{d6.1}%
\end{equation}
for any $a\in{\mathcal{A}},$ being a simple shift on $i\delta\in
i\mathbb{R}\backslash\{0\}\mathbb{\subset C},$ $\ i^{2}=-1,$ $2\pi
/\delta\notin\mathbb{Z}_{+},$ along the complexified loop parameter
$x\in\mathbb{S}^{1}\mathbb{\subset C}.$ The linear and invariant
trace-functional $\tau:{\mathcal{A}}\rightarrow\mathbb{C}$ \ is defined for
any $a\in\mathcal{A}$ by the natural expression :%
\begin{equation}
\tau(a):=\int_{\mathbb{S}^{1}}a(x)dx. \label{d6.1a}%
\end{equation}
Having constructed the basic Lie algebra $\mathcal{G}$ \ of homomorphisms
$A(T)\in Hom{\mathcal{A}},$ where%
\begin{equation}
A(T)\sim\sum^{j\ll\infty}a_{j}(x)T^{j} \label{d6.2}%
\end{equation}
for $a_{j}(y)\in{\mathcal{A}},j\ll\infty.$ As the related Lax type integrable
dynamical systems are generated \cite{Blas,1-Bo,11-BSP,BlPrSa,FaTa,ReSe} by
the Casimir invariants $\gamma\in I(\mathcal{G}^{\ast})$ of the basic Lie
algebra $\mathcal{G},$ satisfying the determining equation \
\begin{equation}
\lbrack\nabla\gamma(l),l]=0, \label{d6.2a}%
\end{equation}
we will be interested in a seed Lax element $l\in\mathcal{G}^{\ast},$ chosen
in the following rationally factorized form:%
\begin{equation}
l:=F_{n}(T)^{-1}\circ Q_{n+p}(T), \label{d6.3}%
\end{equation}
where by definition, the elements \
\begin{equation}
F_{n}(T):=\sum_{j=\overline{0,n}}f_{j}(x)T^{j},\text{ \ \ }Q_{n+p}%
(T):=\sum_{j=\overline{0,n+p}}q_{j}(x)T^{j}\ \label{d6.4}%
\end{equation}
are some polynomial homomorphisms of ${\mathcal{A}}$ for fixed $n$ and
$p\in\mathbb{Z}_{+}.$ \ 

The following \textbf{problem}
\cite{Adler,BlPr-1,BlPrSa,BoLiXi,Dick,Dick-1,HeLe,Kric} arises:
\ \textit{construct the corresponding dynamical systems on the elements}
$F_{n}(T),Q_{n+p}(T)\in Hom{\mathcal{A}},$ \textit{which will possess an
infinite hierarchy of functional invariants and will be analytically
integrable. }

It is natural to consider the general Lax type flow $\ $%
\begin{equation}
dl/dt=[l,\nabla\gamma(l)_{+}], \label{d6.5}%
\end{equation}
for the rational element \ (\ref{d6.3}), generated by a Casimir functional
$\gamma\in$ $I(\mathcal{G}^{\ast})$ and determined by the expression
\ (\ref{d6.2a}). Now let us observe that $\gamma:=tr(\gamma(l)=tr(\gamma
(\tilde{l}))\ \ $for any analytical mapping $\gamma(l)\in\mathcal{G}$ $,$
\ where we have introduced, by definition, the factorized element $\tilde
{l}:=Q_{n+p}F_{n}^{-1}\in\mathcal{G}^{\ast}.$ In addition, the element
$\tilde{l}\ =Q_{n+p}F_{n}^{-1}\in\mathcal{G}^{\ast}$ satisfies the similar to
\ (\ref{d6.5}) evolution equation
\begin{equation}
\text{\ \ \ \ \ \ \ }d\tilde{l}/dt=[\tilde{l},\nabla\gamma(\tilde{l})_{+}]
\label{d6.6}%
\end{equation}
for the same Casimir functional $\gamma\in$ $I(\mathcal{G}^{\ast}),$ whose
gradient, similarly to \ (\ref{d6.2a}), is determined from the algebraic
relationship
\begin{equation}
\lbrack\tilde{l},\nabla\gamma(\tilde{l})]=0.\text{ \ } \label{d6.6a}%
\end{equation}
\ Taking now into account these two compatible equations \ (\ref{d6.5}) and
(\ref{d6.6}) \ one easily derives the following \cite{Dick,Dick-1,Kric}T.
Shiota factorization theorem.

\begin{theorem}
\label{Tm_d6.1} The operator evolution equations
\begin{equation}
dF_{n}/dt=F_{n}\nabla\gamma(l)_{+}-\nabla\gamma(\tilde{l})_{+}F_{n},\text{
\ \ \ \ \ \ \ \ \ \ }dQ_{n+p}/dt=Q_{n+p}\nabla\gamma(l)_{+}-\nabla
\gamma(\tilde{l})_{+}Q_{n+p} \label{d6.7}%
\end{equation}
factorize the Lax type flows \ (\ref{d6.5}) and \ (\ref{d6.6}) for all
$t\in\mathbb{R}$ with elements $l=F_{n}^{-1}Q_{n+p}\in\ \mathcal{G}^{\ast}\ $
and $\tilde{l}=Q_{n+p}F_{n}^{-1}\in\ \mathcal{G}^{\ast},$ respectively,
\ where the corresponding Casimir invariants $\gamma\in I(\mathcal{G}^{\ast
})\ $ satisfy the relationship $\gamma|_{l=F_{n}^{-1}Q_{n+p}}=\gamma
|_{\tilde{l}=Q_{n+p}F_{n}^{-1}}\ $ for any $F_{n}$ and $Q_{n+p}\in\ G_{+}.$
\end{theorem}

As a simple consequence from Theorem \ref{Tm_d6.1} one can derive the
following proposition.

\begin{proposition}
\label{Prop_6.2}There exist such smooth mappings $\Phi,\Psi:$ $\mathbb{R}$
$\rightarrow G\ \ $ to the formal operator subgroup $G\simeq\exp\mathcal{G}$
satisfying the linear evolution equations
\begin{equation}
d\Phi/dt+\nabla\gamma(\tilde{l})_{+}\Phi=0,\text{ \ \ \ \ }d\Psi
/dt+\nabla\gamma(l)_{+}\Psi=0\ \label{d6.8}%
\end{equation}
$\Phi|_{t=0}=\bar{\Phi},B|_{t=0}=\bar{B}\in G,$ \ \ generated, respectively,
by the Lie algebra elements $\nabla\gamma(l)_{+}$ \ and $\ \nabla\gamma
(\tilde{l})_{+}\in\mathcal{G}_{+},$ that
\begin{equation}
F_{n}:=\Psi\text{ }\bar{F}_{n}\Phi^{-1},\text{ \ \ \ \ \ \ \ \ }Q_{n+p}%
:=\Psi\text{ }\bar{Q}_{n+p}\Phi^{-1}, \label{d6.9}%
\end{equation}
where, by definition, the elements $\bar{F}_{n}$ \ and $\bar{Q}_{n+p}\in G$
\ \ are constant with respect to the evolution parameter $t\in\mathbb{R}.$
\end{proposition}

\begin{proof}
It is enough to check, using \ (\ref{d6.8}) that the group elements
\ (\ref{d6.9}) really satisfy the factorized evolution equations \ (\ref{d6.7}).
\end{proof}

Now based on Proposition \ref{Prop_6.2} we can take into account, with no loss
of generality, that the group elements $A,B\in G\ $ for all $t\in\mathbb{R}$
\ can be represented as $\ $\ operator series
\begin{equation}
\Phi(x;t)\sim I+\sum_{j\in\mathbb{Z}_{+}}a_{j}(x;t)T^{-j},\text{ \ \ \ \ }%
\Psi(x;t)\sim I+\sum_{j\in\mathbb{Z}_{+}}b_{j}(x;t)T^{-j},\label{d6.10}%
\end{equation}
whose coefficients can be found recurrently from the expressions
\ (\ref{d6.9}), rewritten in the following useful for calculations form:%
\begin{align}
(I+\sum_{j\in\mathbb{Z}_{+}}b_{j}(x;t))T^{-j})\circ\bar{F}_{n} &  =\ \text{
}F_{n}\circ(I+\sum_{j\in\mathbb{Z}_{+}}a_{j}(x;t)T^{-j}),\text{ \ \ \ }%
\label{d6.11}\\
& \nonumber\\
(I+\sum_{j\in\mathbb{Z}_{+}}b_{j}(x;t)T^{-j})\circ\bar{Q}_{n+p} &  =\ \text{
}Q_{n+p}\circ(I+\sum_{j\in\mathbb{Z}_{+}\backslash\{0\}}a_{j}(x;t)T^{-j}%
),\text{ }\nonumber
\end{align}
where the group elements $\bar{F}_{n}$ and $\bar{Q}_{n+p}\in G_{\ }$ are
considered to be given \textit{a priori \ }constant in the \ following,
motivated by the expression \ (\ref{d6.4}), $\ $\ operator series form:%
\begin{equation}
\bar{F}_{n}(T)\sim\sum_{j\in\mathbb{Z}_{+}}\bar{f}_{j}T^{n-j},\text{ \ \ }%
\bar{Q}_{n+p}(T)\sim\sum_{j\ \in\mathbb{Z}_{+}}\bar{q}_{j}T^{n+p-j}%
\ ,\label{d6.4a}%
\end{equation}
where $d\bar{f}_{j}/dt=0=d\bar{q}_{j}/dt,j\in\mathbb{Z}_{+},$ for all
$\ t\in\mathbb{R}.$ \ The results obtained above mean, in particular, that the
expressions\textit{ }\ (\ref{d6.11}) can be effectively used for finding exact
analytical solutions to the resulting differential-functional equations
naturally following from the operator evolution equations \ (\ref{d6.7}),
generated by a suitably chosen Casimir functional $\gamma\in$ $\mathrm{I}%
(\mathcal{G}^{\ast}).$ This and other related aspects of this important
problem of finding exact analytical solutions will be analyzed in detail in
other work under preparation.

\section{\label{Sec_2}The centrally extended basic associative algebra case}

Consider now the case when the basic associative functional algebra
${\mathcal{A}}\subset C^{\infty}(\mathbb{S}^{1};\mathbb{C})$ is extended as
the $\mathtt{loop}$ algebra $C_{S^{1}}({\mathcal{A)}}$ of smooth mappings
$\{\mathbb{S}^{1}\rightarrow{\mathcal{A}}\}.$ The corresponding Lie algebra
$C_{S^{1}}(\mathcal{G}{\mathcal{)}}$ of linear homomorphisms $C_{S^{1}%
}(\mathcal{A}{\mathcal{)}},$ naturally generated by the complexified
homomorphic shifts \ (\ref{d6.1}) along the cyclic variable $x\in
\mathbb{S}^{1},$ can be centrally extended to the Lie algebra $C_{S^{1}%
}(\widehat{\mathcal{G}}{\mathcal{)}}$ via the standard
\cite{FaTa,BlPrSa,PHS,ReSe} Maurer-Cartan cocycle \
\begin{equation}
\omega_{2}(A(T),B(T)):=\int_{{\mathbb{S}}^{1}}dy\,\langle A(T),dB(T)/dy\rangle
, \label{d6.11a}%
\end{equation}
which also admits the natural splitting subject to the positive and negative
degrees of the basic homomorphism \ (\ref{d6.1}) into two subalgebras
\begin{equation}
C_{\mathbb{S}^{1}}(\widehat{\mathcal{G}})=C_{\mathbb{S}^{1}}(\widehat
{\mathcal{G}}{\mathcal{)}}_{+}+C_{\mathbb{S}^{1}}(\widehat{\mathcal{G}%
}{\mathcal{)}}_{-}. \label{d6.12}%
\end{equation}
The latter makes it possible to construct the adjoint splitting
\begin{equation}
C_{\mathbb{S}^{1}}(\widehat{\mathcal{G}}^{\ast})=C_{\mathbb{S}^{1}}%
(\widehat{\mathcal{G}}^{\ast}{\mathcal{)}}_{+}+C_{\mathbb{S}^{1}}%
(\widehat{\mathcal{G}}^{\ast}{\mathcal{)}}_{-} \label{d6.13}%
\end{equation}
and define for any factorized element $(l,1)\in C_{\mathbb{S}^{1}}%
(\widehat{\mathcal{G}}^{\ast})$ the following \ integrable Hamiltonian flows%
\begin{equation}
dl/dt=[l-d/dy,\nabla\gamma(l)_{+}],\text{ \ } \label{d6.14}%
\end{equation}
where the Casimir functionals $\gamma\in I(C_{\mathbb{S}^{1}}(\widehat
{\mathcal{G}}^{\ast}))$ satisfy the gauge type differential-functional
equation
\begin{equation}
\lbrack l,\nabla\gamma(l)]=d\nabla\gamma(l)/dy\ \label{d6.15}%
\end{equation}
for all $y\in\mathbb{S}^{1}.$ Here as above we will consider the case when a
rationally factorized element $l(T)\in C_{\mathbb{S}^{1}}(\mathcal{G}^{\ast})$
is given in the form \
\begin{equation}
l(T):=F_{n\ }(T)^{-1}\circ Q_{n+p}(T), \label{d6.16}%
\end{equation}
where, by definition, the elements \
\begin{equation}
F_{n}(T):=\sum_{j=\overline{0,n\ }}f_{j}(x;y)T^{j},\text{ \ \ }Q_{n+p}%
(T):=\sum_{j=\overline{0,n}+p}q_{j}(x;y)T^{j}\ \label{d6.17}%
\end{equation}
belong to the formal operator subgroup $C_{\mathbb{S}^{1}}(G_{+}%
):=\exp(C_{\mathbb{S}^{1}}(\mathcal{G}_{+}))\simeq I+C_{\mathbb{S}^{1}%
}(\mathcal{G}_{+}).$

Inasmuch in this case we also \ can not make use of the $\ $\ expansions
\ (\ref{d6.10}), thus forcing us \ to apply the Lie-algebraic scheme of
\cite{BlPr-1,BlPrSa,ReSe-1}. Namely, we will formulate the following similar
statements without proof.

\begin{lemma}
\label{Lm_d.6.1}For any\ factorized in the rational form \ (\ref{d6.16})
element $l\in C_{\mathbb{S}^{1}}(\mathcal{G}^{\ast})$ \ there exists, as
$2\pi/\delta\notin\mathbb{Z}_{+},$ an invertible mapping $\Phi(T)\in
C_{\mathbb{S}^{1}}(G_{-}),$ $\Phi(T)|_{x=0}=I,$ and such an element $\bar
{l}\in C_{\mathbb{S}^{1}}({\mathcal{G}}^{\ast})$ that the following functional
operator relationship
\begin{equation}
(\partial/\partial y-l(T))\circ\Phi(T)=\Phi(T)\circ(\partial/\partial
y-\bar{l}(T)\ )\ \label{d6.18}%
\end{equation}
holds, where $\partial\bar{l}(T)/\partial t=0=$ $\partial\bar{l}(T)/\partial
x,$ \ that is the element $\bar{l}(T)\in C_{\mathbb{S}^{1}}({\mathcal{G}%
}^{\ast})$ is constant \ both with respect to the evolution parameter
$t\in\mathbb{R}$ and the functional algebra $\mathcal{A}$ parameter
$x\in\mathbb{S}^{1}.$
\end{lemma}

\begin{proof}
\textit{Sketch of a proof.} Taking into account that $l(T):=F_{n\ }%
(T)^{-1}\circ Q_{n+p}(T)\in C_{\mathbb{S}^{1}}({\mathcal{G}}^{\ast}),$ the
operator relationship \ (\ref{d6.18}) can be equivalently rewritten as
\begin{equation}
(F_{n\ }(T)\partial/\partial y-Q_{n+p}(T))\circ\Phi(T)=F_{n\ }(T)\circ
\Phi(T)\circ(\partial/\partial y-\bar{l}(T)), \label{d6.18a}%
\end{equation}
where
\begin{equation}
\bar{l}(T)\sim\sum_{j\in\mathbb{Z}_{+}}l_{j}(y)T^{p-j} \label{d6.18ab}%
\end{equation}
is constant, it allows to determine recurrently all \ coefficients of the
corresponding invertible operator expansion
\begin{equation}
\Phi(T)\sim I+\sum_{j\in\mathbb{Z}_{+}}\varphi_{j}(x;y)T^{-j} \label{d6.18b}%
\end{equation}
for all $(x;y)\in{\mathbb{S}}^{1}\times{\mathbb{S}}^{1}.$ The latter proves
the lemma.
\end{proof}

\begin{theorem}
\label{Tm_d6.2} The following functionals
\begin{equation}
\gamma_{j}=Tr(T^{j}\bar{l}(T))=\int_{{\mathbb{S}}^{1}}tr(\bar{l}%
_{j}(y))dy=\int_{{\mathbb{S}}^{1}}\tau(\bar{l}_{j}(y))dy, \label{d6.19}%
\end{equation}
where, by definition,
\begin{equation}
\bar{l}\sim\sum_{j\in\mathbb{Z}_{+}}\bar{l}_{j}(y)T^{p-j}, \label{d6.19a}%
\end{equation}
are for all $j\in$ $\mathbb{Z}_{+}$ the Casimir invariants for the centrally
extended loop Lie algebra $C_{\mathbb{S}^{1}}(\widehat{\mathcal{G}}).$
\end{theorem}

Based on Theorem \ref{Tm_d6.2} \ one can find that the corresponding
gradients
\begin{equation}
\nabla\gamma_{j}(l)=\Phi(T)T^{j}\Phi(T)^{-1}, \label{d6.20}%
\end{equation}
$j\in$ $\mathbb{Z}_{+},$ for the countable hierarchy of Casimir functionals
\ \ (\ref{d6.19}) satisfy the determining relationship \ (\ref{d6.15}). \ In
addition, from \ (\ref{d6.18}) one ensues that the following \ operator
expression \ for $\ $the case $l=F_{n}^{-1}Q_{n+p}\in C_{\mathbb{S}^{1}%
}({\mathcal{G}}^{\ast}):$%
\begin{equation}
\bar{l}=\Phi(T)^{-1}(l-\partial/\partial y)\Phi(T),\ \label{d6.21}%
\end{equation}
holds for all $y\in\mathbb{S}^{1},$ where the invertible mapping $\Phi(T)\in
C_{\mathbb{S}^{1}}(\hat{G})$ satisfies the evolution equation
\begin{equation}
d\Phi(T)/dt+\nabla\gamma(l)_{+}\Phi(T)=0\ \ \label{d6.25}%
\end{equation}
for all $t\in\mathbb{R}.$ Similarly one can state that there exists a suitably
chosen mapping $\Psi(T)\in C_{\mathbb{S}^{1}}(G)\ $for the case $\tilde
{l}=Q_{n+p}F_{n}^{-1}\in C_{\mathbb{S}^{1}}(\mathcal{G}^{\ast}),\ $ such that
\
\begin{equation}
\overline{\tilde{l}}=\Psi(T)^{-1}(\tilde{l}-\partial/\partial y)\Psi
(T)\ \ \label{d6.22}%
\end{equation}
holds for some constant element $\widetilde{\bar{l}}\in C_{\mathbb{S}^{1}%
}(\mathcal{G}^{\ast})\ $with respect to both the evolution variable
$t\in\mathbb{R}$ and the functional parameter $x\in\mathbb{S}^{1},$ \ where
the invertible mapping $\Psi(T)\in C_{\mathbb{S}^{1}}(G)$ satisfies the
evolution equation
\begin{equation}
d\Psi(T)/dt+\nabla\gamma(l)_{+}\Psi(T)=0.\ \ \label{d6.22a}%
\end{equation}
Moreover, the element $\tilde{l}\in C_{\mathbb{S}^{1}}(\mathcal{G}^{\ast})$
satisfies the Lax type evolution equation
\begin{equation}
d\tilde{l}/dt=[\tilde{l}-d/dy,\nabla\gamma(\tilde{l})_{+}]\ \text{ \ }
\label{d6.22ab}%
\end{equation}
for all $t\in\mathbb{R}.$ Taking into account that the expression
\ (\ref{d6.22}) can be equivalently rewritten as
\begin{equation}
(F_{n\ }(T)\partial/\partial y-Q_{n+p}(T)\ )\circ\Psi(T)=F_{n\ }(T)\circ
\Psi(T)\circ(\partial/\partial y-\overline{\tilde{l}}(T)), \label{d6.22b}%
\end{equation}
from \ (\ref{d6.22b}), \ (\ref{d6.18a}) and evolution equations \ (\ref{d6.25}%
), \ (\ref{d6.22a}) one can derive the corresponding factorized evolution
equations
\begin{equation}
d\mathrm{F}_{n}/dt=\mathrm{F}_{n}\nabla\gamma(l)_{+}-\nabla\gamma(\tilde
{l})_{+}\mathrm{F}_{n},\text{ \ }\mathrm{Q}_{n}/dt=\mathrm{Q}_{n}\nabla
\gamma(l)_{+}-\nabla\gamma(\tilde{l})_{+}\mathrm{Q}_{n}, \label{d6.23}%
\end{equation}
for the elements $\mathrm{F}_{n}:=F_{n}\in$ $C_{\mathbb{S}^{1}}(G)$ and
$\mathrm{Q}_{n+p}:=Q_{n+p}-F_{n}\partial/\partial y\in C_{\mathbb{S}^{1}}%
(\hat{G}),$ \ which allow the following natural representations
\begin{equation}
\mathrm{F}_{n}:=\Psi(T)\text{ }\mathrm{\bar{F}}_{n}\Phi(T)^{-1},\text{
\ \ \ \ \ \ \ \ }\mathrm{Q}_{n+p}:=\Psi(T)\text{ }\mathrm{\bar{Q}}_{n+p}%
\Phi(T)^{-1}\ \label{d6.24}%
\end{equation}
with $\mathrm{\bar{F}}_{n}:=\bar{F}_{n}\in C_{\mathbb{S}^{1}}(G)$ and
$\mathrm{\bar{Q}}_{n}:=\bar{Q}_{n}-\bar{F}_{n}\partial/\partial y\in
C_{\mathbb{S}^{1}}(G)$ being constants with respect to the evolution variables
$t\in\mathbb{R}$ and $x\in\mathbb{S}^{1}.$ \ Taking now into account the above
expressions \ (\ref{d6.24}) and \ (\ref{d6.25}) one easily obtains from
\ (\ref{d6.23}) the following evolutions equations
\begin{align}
dF_{n}/dt  &  =F_{n}\nabla\gamma(l)_{+}-\nabla\gamma(\tilde{l})_{+}%
F_{n},\text{ \ \ \ }\label{d6.26}\\
& \nonumber\\
dQ_{n+p}/dt  &  =Q_{n+p}\nabla\gamma(l)_{+}-\nabla\gamma(\tilde{l})_{+}%
Q_{n+p}-F_{n}\partial\nabla\gamma(l)_{+}/\partial y\nonumber
\end{align}
on the basic operator factors $F_{n}$ and $Q_{n+p}\in C_{\mathbb{S}^{1}}(G).$
The corresponding invertible mapping $\Phi$ and $\Psi\in C_{\mathbb{S}^{1}%
}(G),$ satisfying, respectively, the expressions\ (\ref{d6.24}), \ can be
recurrently constructed from the \ algebraic relationships
\begin{align}
F_{n}(I+\sum_{j\in\mathbb{Z}_{+}}a_{j}(x;y)T^{-j})  &  =(I+\sum_{j\in
\mathbb{Z}_{+}}b_{j}(x;y)T^{-j})\text{ }\bar{F}_{n},\ \label{d6.27}\\
& \nonumber\\
(Q_{n+p}-F_{n}\partial/\partial y)(I+\sum_{j\in\mathbb{Z}_{+}}a_{j}%
(x;y)T^{-j})  &  =(I+\sum_{j\in\mathbb{Z}_{+}}b_{j}(x;y)T^{-j})\text{ }%
(\bar{Q}_{n}-\bar{F}_{n}\partial/\partial y)\nonumber
\end{align}
in the series expansion form:
\begin{equation}
\Phi(T)\sim I+\sum_{j\in\mathbb{Z}_{+}}a_{j}(x;y)T^{-j},\text{ \ \ \ \ }%
\Psi(T)\sim I+\sum_{j\in\mathbb{Z}_{+}}b_{j}(x;y)T^{-j}. \label{d6.28}%
\end{equation}
The statements above we can formulate as the next generalized T. Shiota type
factorization theorem.

\begin{theorem}
\label{Tm_d6.3} The operator evolution equations%
\begin{align}
dF_{n}/dt  &  =F_{n}\nabla\gamma(l)_{+}-\nabla\gamma(\tilde{l})_{+}%
F_{n},\text{ \ \ \ }\label{d6.29}\\
& \nonumber\\
dQ_{n+p}/dt  &  =Q_{n+p}\nabla\gamma(l)_{+}-\nabla\gamma(\tilde{l})_{+}%
Q_{n+p}-F_{n}\partial\nabla\gamma(l)_{+}/\partial y\nonumber
\end{align}
factorize the Lax type flows \ (\ref{d6.14}) and \ (\ref{d6.22ab}) with
elements $l=F_{n}^{-1}Q_{n+p}\in C_{\mathbb{S}^{1}}(\mathcal{G}^{\ast})$ and
$\tilde{l}=Q_{n+p}F_{n}^{-1}\in C_{\mathbb{S}^{1}}(\mathcal{G}^{\ast}),$
respectively, where the corresponding Casimir invariants $\gamma\in
I(C_{\mathbb{S}^{1}}(\widehat{\mathcal{G}}^{\ast}))\ $ satisfy the
relationship $\gamma|_{l=F_{n}^{-1}Q_{n+p}}=\gamma|_{\tilde{l}=Q_{n+p}%
F_{n}^{-1}}\ $ for any $F_{n}$ and $Q_{n+p}\in C_{\mathbb{S}^{1}}(G_{+}).$
\end{theorem}

\subsection{Example}

As an example of a rationally factorized operator $l\in C_{\mathbb{S}^{1}%
}(\mathcal{G}^{\ast})$ one can consider the following simple expression
\begin{equation}
l:=T^{-1}(T^{2}+Tv+uI), \label{d6.30}%
\end{equation}
where functions $u,v$ $\in C(\mathbb{S}^{1}\times\mathbb{R};\mathbb{R}).$ The
corresponding \ elements $F_{1}:=T,$ $Q_{2}:=T^{2}+\ Tv+uI)\in C_{\mathbb{S}%
^{1}}(G_{+})$ generate the factorized evolution equations \ (\ref{d6.29}),
where gradients of the corresponding Casimir functionals $\gamma\in
I(C_{\mathbb{S}^{1}}(\widehat{\mathcal{G}}^{\ast}))$ can be found recurrently
from the relationships \ (\ref{d6.20}) jointly with the relationships
\ (\ref{d6.18a}), \ (\ref{d6.18ab}) and \ (\ref{d6.18b}). From the
corresponding calculations one ensues the system of integrable evolution
functional equations
\begin{equation}
u_{t}=u(Tv-v),v_{t}=v(T^{-1}u-u)\ \label{d6.31}%
\end{equation}
on the elements $\ u,v\ \in C(\mathbb{S}^{1};\mathbb{R}^{\ }).$

\section{\label{Sec_3}Special functional-algebraic realizations}

The algebraic scheme devised in Section \ref{Sec_2} makes it possible to be
effectively modified for the case when the associative functional algebra
${\mathcal{A}}$ is chosen to be the algebra of smooth pseudo-differential
operators $\ \mathrm{PDO}(\mathbb{S}^{1}),\ $acting on the functional space
$C^{\infty}(\mathbb{S}^{1};\mathbb{R}^{\ })\ $and endowed with the natural
commutator Lie structure. The resulting Lie algebra $\mathcal{G}%
:=(\mathrm{PDO}(\mathbb{S}^{1});[\cdot,\cdot])$ is split into direct sum of
two subalgebras, $\mathcal{G=G}_{+}\oplus\mathcal{G}_{-}:$%
\begin{align}
\mathcal{G}_{+}  &  :=\{\sum_{0\leq j\ll\infty}a_{j}(x)\partial^{j}%
:a_{j}(x)\ \in C^{\infty}(\mathbb{S}^{1};\mathbb{R}^{\ }),0\leq j\ll
\infty\},\label{d7.1}\\
& \nonumber\\
\mathcal{G}_{-}  &  :=\{\sum_{j\in\mathbb{Z}_{+}}b_{j}(x)\partial^{-j}%
:b_{j}(x)\ \in C^{\infty}(\mathbb{S}^{1};\mathbb{R}^{\ }),0\leq j\ll
\infty\},\nonumber
\end{align}
where, by definition, $\partial:=\partial/\partial x$ and $\partial
\cdot\partial^{-1}=1\ \ \ $for \ $x\in\mathbb{S}^{1}.$ Moreover, the Lie
algebra $\mathcal{G}$ is metrized by means of the invariant trace form
\begin{equation}
(a,b):=Tr(a\cdot b),\ Tr(c):=\int_{\mathbb{S}^{1}}(\mathrm{res}_{\partial
}\text{ }c)\text{ }dx \label{d7.2}%
\end{equation}
for any $a,b$ and $c\in\mathcal{G},$ allowing to identify the adjoint space
$\mathcal{G}^{\ast}\simeq\mathcal{G}.$ \ 

Taking into account these preliminaries a \ similar to that, posed in Section
\ref{Sec_2}, problem arises: \ \textit{construct the corresponding operator
dynamical systems on the elements} $F_{n}(\partial),Q_{n+p}(\partial
)\in\mathcal{G},$ \textit{which will possess an infinite hierarchy of
functional invariants and will be analytically integrable. }

As above we consider the general Lax type flow $\ $%
\begin{equation}
dl/dt=[l,\nabla\gamma(l)_{+}], \label{d7.3}%
\end{equation}
for the rational element \
\begin{equation}
l(\partial):=F_{n}(\partial)^{-1}Q_{n+p}(\partial), \label{d7.4}%
\end{equation}
generated by a Casimir functional $\gamma\in$ $I(\mathcal{G}^{\ast})$ and
determined by the expression \ (\ref{d6.2a}). \ One observes that
$\gamma:=tr(\gamma(l)=tr(\gamma(\tilde{l}))\ \ $for any analytical mapping
$\gamma(l)\in\mathcal{G}$ $,$ \ where we have introduced, by definition, the
factorized element $\tilde{l}:=Q_{n+p}F_{n}^{-1}\in\mathcal{G}^{\ast}.$ Also
the element $\tilde{l}\ =Q_{n+p}F_{n}^{-1}\in\mathcal{G}^{\ast}$ satisfies the
similar to \ (\ref{d6.5}) evolution equation
\begin{equation}
\text{\ \ \ \ \ \ \ }d\tilde{l}/dt=[\tilde{l},\nabla\gamma(\tilde{l})_{+}]
\label{d7.5}%
\end{equation}
for the same Casimir functional $\gamma\in$ $I(\mathcal{G}^{\ast}),$ whose
gradient, similarly to \ (\ref{d6.2a}), is determined from the algebraic
relationship
\begin{equation}
\lbrack\tilde{l},\nabla\gamma(\tilde{l})]=0.\text{ \ } \label{d7.6}%
\end{equation}
\ Taking now into account these two compatible equations \ (\ref{d7.3}) and
(\ref{d7.5}) \ one easily derives the following factorization theorem.

\begin{theorem}
\label{Tm_d7.1} The differential operator evolution equations
\begin{equation}
dF_{n}/dt=F_{n}\nabla\gamma(l)_{+}-\nabla\gamma(\tilde{l})_{+}F_{n},\text{
\ \ \ \ \ \ \ \ \ \ }dQ_{n+p}/dt=Q_{n+p}\nabla\gamma(l)_{+}-\nabla
\gamma(\tilde{l})_{+}Q_{n+p} \label{d7.7}%
\end{equation}
factorize the Lax type flows \ (\ref{d7.3}) and \ (\ref{d7.5}) for all
$t\in\mathbb{R}$ with elements $l=F_{n}^{-1}Q_{n+p}\in\ \mathcal{G}^{\ast}\ $
and $\tilde{l}=Q_{n+p}F_{n}^{-1}\in\ \mathcal{G}^{\ast},$ respectively,
\ where the corresponding Casimir invariants $\gamma\in I(\mathcal{G}^{\ast
})\ $ satisfy the relationship $\gamma|_{l=F_{n}^{-1}Q_{n+p}}=\gamma
|_{\tilde{l}=Q_{n+p}F_{n}^{-1}}\ $ for any $F_{n}$ and $Q_{n+p}\in\ G_{+}.$
\end{theorem}

From Theorem \ref{Tm_d7.1} one easily ensues the following proposition.

\begin{proposition}
\label{Prop_7.2}There exist such smooth mappings $\Phi,\Psi:$ $\mathbb{R}$
$\rightarrow G\ \ $ to the formal operator subgroup $G\simeq\exp\mathcal{G}$
satisfying the linear evolution equations
\begin{equation}
d\Phi/dt+\nabla\gamma(\tilde{l})_{+}\Phi=0,\text{ \ \ \ \ }d\Psi
/dt+\nabla\gamma(l)_{+}\Psi=0\ \label{d7.8}%
\end{equation}
$\Phi|_{t=0}=\bar{\Phi},B|_{t=0}=\bar{B}\in G,$ \ \ generated, respectively,
by the pseudo-differential Lie algebra elements $\nabla\gamma(l)_{+}$ \ and
$\ \nabla\gamma(\tilde{l})_{+}\in\mathcal{G}_{+},$ that
\begin{equation}
F_{n}:=\Psi\text{ }\bar{F}_{n}\Phi^{-1},\text{ \ \ \ \ \ \ \ \ }Q_{n+p}%
:=\Psi\text{ }\bar{Q}_{n+p}\Phi^{-1}, \label{d7.9}%
\end{equation}
where, by definition, the elements $\bar{F}_{n}$ \ and $\bar{Q}_{n+p}\in G$
\ \ are some constant expressions with respect to the evolution parameter
$t\in\mathbb{R}.$
\end{proposition}

\section{\label{Sec_4}The Poisson structures and Hamiltonian analysis on the
extended phase space}

Let us consider \ \ equation \ (\ref{d6.5}), the first equation of
(\ref{d6.8}) and its adjoint expression:
\begin{equation}
d\hat{l}/dt=[\hat{l},\nabla\gamma(\hat{l})_{+}],\text{ \ \ }d\hat{f}%
/dt+\nabla\gamma(\hat{l})_{+}\hat{f}=0,\text{ \ \ \ \ }d\hat{f}^{\ast
}/dt-\nabla\gamma(\hat{l})_{+}^{\ast}\hat{f}=0\ \label{ce15a}%
\end{equation}
for \ vector elements $\hat{f}\in W$ and $\hat{f}^{\ast}\in W^{\ast},$
respectively, \ where $\ W$ denotes a representation space for the group $G$
and $W^{\ast}$ is its natural conjugation with respect to the natural bilinear
form $\ <\cdot,\cdot>,$ \ realizing the standard paring between spaces
$W^{\ast}$ and $W.$ Put also by
\[
\nabla\gamma(\hat{l},\hat{f},\hat{f}^{\ast}):=(\delta\gamma/\delta\hat
{l},\,\delta\gamma/\delta\hat{f},\,\delta\gamma/\delta\hat{f}^{\ast
})^{\intercal}%
\]
an extended gradient vector at a point $(\hat{l};\tilde{f},\tilde{f}^{\ast
})^{\intercal}\in\mathcal{G}^{\ast}\oplus W\oplus W^{\ast}$ for any smooth
functional $\gamma\in\mathcal{D}(\mathcal{G}^{\ast}\oplus W\oplus W^{\ast}).$

On the space $\mathcal{G}^{\ast}$ there exists the canonical Poisson
structure
\begin{equation}
\delta\gamma/\delta\hat{l}:\overset{\tilde{\theta}}{\rightarrow}[\hat
{l},(\delta\gamma/\delta\hat{l})_{+}]-[\hat{l},\delta\gamma/\delta\hat{l}%
]_{+}\ , \label{ce15}%
\end{equation}
where $\tilde{\theta}:T^{\ast}(\mathcal{G}^{\ast})\rightarrow T(\mathcal{G}%
^{\ast})\simeq\mathcal{G}$ is a Poisson operator at a point $\hat{l}%
\in\mathcal{G}^{\ast}.\ $Similarly on the space $W\oplus W^{\ast}$ \ there
exists the canonical Poisson structure
\begin{equation}
(\delta\gamma/\delta\hat{f},\,\delta\gamma/\delta\hat{f}^{\ast})^{\intercal
}:\overset{\tilde{J}}{\rightarrow}(-\delta\gamma/\delta\hat{f}^{\ast}%
,\,\delta\gamma/\delta\hat{f})^{\mathbb{\intercal}}\ , \label{ce16}%
\end{equation}
where $\tilde{J}:T^{\ast}(W\oplus W^{\ast})\rightarrow T(W\oplus W^{\ast})$ is
the Poisson operator corresponding to the symplectic form $\omega^{(2)}%
=<d\hat{f}^{\ast},\wedge d\hat{f}$ $>$\ \ at a point $(\hat{f},\hat{f}^{\ast
})\in W\oplus W^{\ast}.$ \ \ It should be noted \ here that the Poisson
structure (\ref{ce15}) generates equations ~(\ref{d6.5}) and \ (\ref{d6.6})
for any Casimir functional $\gamma\in\mathrm{I}(\mathcal{G}^{\ast}).$

Thus, on the extended phase space $\mathcal{G}^{\ast}\oplus W\oplus W^{\ast}$
one can obtain a new Poisson structure as the tensor product $\tilde{\Theta
}:=\tilde{\theta}\otimes\tilde{J}$ \ of the structures (\ref{ce15}) and
(\ref{ce16}).

Consider now the following Backlund \cite{BlPrSa} transformation:
\begin{equation}
(\hat{l};\hat{f},\hat{f}^{\ast})^{\mathbb{\intercal}}:\overset{B}{\rightarrow
}(l=\ l(\hat{l};\hat{f},\hat{f}^{\ast}),\text{ \ \ }f=\hat{f},\text{
\ \ }f^{\ast}=\hat{f}^{\ast})^{\mathbb{\intercal}}, \label{ce17}%
\end{equation}
generating on $\mathcal{G}^{\ast}\oplus W\oplus W^{\ast}$ \ some Poisson
structure $\Theta:T^{\ast}(\mathcal{G}^{\ast}\oplus W\oplus W^{\ast
})\rightarrow T(\mathcal{G}^{\ast}\oplus W\oplus W^{\ast}).$ The \textit{main
condition imposed on the mapping (\ref{ce17}) is the coincidence of the
resulting dynamical system }%
\begin{equation}
(dl/dt;\,df/dt,\,df^{\ast}/dt)^{\mathbb{\intercal}}:=-\Theta\nabla\bar{\gamma
}(\ l;f,f^{\ast}) \label{ce18}%
\end{equation}
with the evolution equations
\begin{equation}
dl/dt=[l,\nabla\gamma(l)_{+}],\text{ \ \ }df/dt=\nabla\gamma\ (l)_{+}%
f,\ \text{\ \ }df^{\ast}/dt=-\nabla\gamma\ (l)_{+}f^{\ast}\ \ \label{ce13}%
\end{equation}
in the case when $\bar{\gamma}:=\gamma\in\mathrm{I}(\mathcal{G}^{\ast}),$
being not dependent on the variables $(f,f^{\ast})\in W\oplus W^{\ast}.$

To satisfy that condition we will find variation of the functional
$\bar{\gamma}:=\gamma|_{l=l(\hat{l},\hat{f},\hat{f}^{\ast})}\in\mathcal{D}%
(\mathcal{G}^{\ast}\times W\oplus W^{\ast}),$ generated by a Casimir
functional $\gamma\in\mathrm{I}(\mathcal{G}^{\ast}),$ under the constraint
$\delta\tilde{l}=0,$ taking into account the evolutions (\ref{ce15a}) and the
Backlund transformation (\ref{ce17}) definition. One easily obtains that
\begin{gather}
\left.  \delta\bar{\gamma}(\hat{l};\hat{f},\hat{f}^{\ast})\right\vert
_{\delta\hat{l}=0}\qquad{}=<\delta\bar{\gamma}/\delta\hat{f},\delta\hat
{f}>+<\delta\bar{\gamma}/\delta\hat{f}^{\ast},\delta\hat{f}^{\ast}>\nonumber\\
\qquad{}=\left.  <-d\hat{f}^{\ast}/dt,\delta\hat{f}>+<d\hat{f}/dt,\delta
\hat{f}^{\ast}>\right\vert _{\hat{f}=f,\,\hat{f}^{\ast}=f^{\ast}}=\nonumber\\
\qquad{}=<(\delta\gamma/\delta\ l)_{+}^{\ast}\hat{f}^{\ast},\delta\hat
{f}>+<(\delta\gamma/\delta\ l)_{+}\hat{f},\delta\hat{f}^{\ast}>=\nonumber\\
\qquad{}=<\hat{f}^{\ast},(\delta\gamma/\delta\ l)_{+}\delta\hat{f}%
>+<(\delta\gamma/\delta\ l)_{+}\hat{f},\delta\hat{f}^{\ast}>=\nonumber\\
\qquad{}=(\delta\gamma/\delta\ l,(\delta\hat{f})\xi^{-1}\otimes\hat{f}^{\ast
})+(\delta\gamma/\delta\ l,\hat{f}\xi^{-1}\otimes\delta\hat{f}^{\ast
})=\nonumber\\
\qquad{}=\left(  \delta\gamma/\delta\ l,\delta(\hat{f}\xi^{-1}\otimes\hat
{f}^{\ast})\right)  :=(\delta\gamma/\delta\ l,\delta\ l)\ , \label{ce19}%
\end{gather}
giving rise to the relationship
\begin{equation}
\left.  \delta\ l\right\vert _{\delta\tilde{l}=0}=\delta(\hat{f}\xi
^{-1}\otimes\hat{f}^{\ast}):=\ \delta(\hat{f}\xi^{-1}\otimes\hat{f}^{\ast}).
\label{ce19a}%
\end{equation}
Having assumed now the linear dependence of $\ l$ on $\tilde{l}\in
\hat{\mathcal{G}}^{\ast}$ one gets right away from \ (\ref{ce19a}) that
%%%20%
\begin{equation}
\ l=\tilde{l}+\hat{f}\xi^{-1}\otimes\hat{f}^{\ast}. \label{ce20}%
\end{equation}
Thus, the Backlund transformation (\ref{ce17}) can be rewritten as
\begin{equation}
(\hat{l};\hat{f},\hat{f}^{\ast})^{\mathbb{\intercal}}:\overset{B}{\rightarrow
}(\ l=\hat{l}+\hat{f}_{\ }\xi^{-1}\otimes\hat{f}^{\ast};\text{ \ }f\text{
}=\hat{f}\ ,\text{ }f^{\ast}=\hat{f}^{\ast})^{\mathbb{\intercal}}.
\label{ce21}%
\end{equation}
Now by means of simple calculations via \cite{BlPrSa} the isomorphism formula
\[
\Theta=B^{^{\prime}}\tilde{\Theta}B^{^{\prime}\ast}\ ,
\]
where $B^{^{\prime}}:T(\mathcal{G}^{\ast}\oplus W\oplus W^{\ast})\rightarrow
T(\mathcal{G}^{\ast}\oplus W\oplus W^{\ast})$ is a Frechet derivative of
(\ref{ce21}), one finds easily the following form of the Backlund transformed
Poisson structure $\Theta$ on $\mathcal{G}^{\ast}\oplus W\oplus W^{\ast}:$
%%%22%
\begin{equation}
\Theta:\nabla\gamma(\ l;f,f^{\ast})\rightarrow\left(
\begin{array}
[c]{c}%
\left[  \ l,(\delta\gamma/\delta\ l)_{+}\right]  -\left[  \ l,\delta
\gamma/\delta\ l\right]  _{+}\\
+(f\xi^{-1}\otimes(\delta\gamma/\delta f)-(\delta\gamma/\delta f^{\ast}%
)\xi^{-1}\otimes f^{\ast});\\
-\delta\gamma/\delta f^{\ast}-(\delta\gamma/\delta\ l)_{+}f\\
\delta\gamma/\delta f+(\delta\gamma/\delta\ l)_{+}^{\ast}f^{\ast}%
\end{array}
\right)  ^{\intercal}\ , \label{ce22}%
\end{equation}
where $\gamma\in D(\mathcal{G}^{\ast}\oplus W\oplus W^{\ast})$ is an arbitrary
smooth functional. \ The obtained Backlund transformation (\ref{ce21}) makes
it possible to formulate the following theorem.

\begin{theorem}
The set of differential-operator dynamical systems (\ref{ce13}) on
$\mathcal{G}^{\ast}\oplus W\oplus W^{\ast}$ is Hamiltonian with respect to the
Poisson structure \ (\ref{ce22}) and has the form \ (\ref{ce18}) for
$\gamma:=\bar{\gamma}\in\mathrm{I}(\mathcal{G}^{\ast}),$ \ \ being chosen
Casimir functionals on $\mathcal{G}^{\ast}.$
\end{theorem}

Based on the expression (\ref{ce18}) one can construct a new hierarchy of
Hamiltonian evolution equations describing commutative flows generated by
involutive with respect to the Poisson bracket (\ref{ce22}) Casimir invariants
$\gamma\in\mathrm{I}(\hat{\mathcal{G}}^{\ast}),$ extended on the space
$\mathcal{G}^{\ast}\oplus W\oplus W^{\ast}.$

\ Proceed now to considering flows (\ref{d6.5}) and (\ref{d6.6}) as
Hamiltonian systems on $\mathcal{G}^{\ast}\times\mathcal{G}^{\ast}$ subject to
the following tonsor doubled standard Poisson structure:
\begin{equation}
\vartheta:\nabla\gamma(l)\longrightarrow\binom{\left[  \nabla\gamma
(l)_{+},l\right]  {\LARGE -}\left[  \nabla\gamma(l),l\right]  _{+}}{\left[
\nabla\tilde{\gamma}_{+}(\tilde{l}),\tilde{l}\right]  -\left[  \nabla
\tilde{\gamma}(\tilde{l}),\tilde{l}\right]  _{+}}, \label{3.1}%
\end{equation}
where $\gamma(l)=\gamma(\tilde{l})$ and $\ \gamma\in\mathcal{D}\mathfrak{(}%
\mathcal{G}^{\ast}\times\mathcal{G}^{\ast})$ is an arbitrary smooth functional
on $\mathcal{G}^{\ast}\times\mathcal{G}^{\ast}.$ Concerning the transformation%

\begin{equation}
\Phi(Q,F;\tilde{l},l)=0\Leftrightarrow\tilde{l}-QF^{-1}=0,\ \ l-F^{-1}Q=0,
\label{3.2}%
\end{equation}
which can be evidently considered as a usual Backlund \cite{BlPrSa}
transformation, we can construct a new Poisson structure $\eta:T^{\ast
}(\mathcal{G}_{+}^{\ast}\times\mathcal{G}_{+}^{\ast})\longrightarrow
T(\mathcal{G}_{+}^{\ast}\times\mathcal{G}_{+}^{\ast})$ on the space
$\mathcal{G}_{+}^{\ast}\times\mathcal{G}_{+}^{\ast}$ \ with respect to the
phase variables $(F,Q)\in\mathcal{G}_{+}^{\ast}\times\mathcal{G}_{+}^{\ast}.$
Thereby one finds \cite{BlPrSa} the corresponding to (\ref{3.1}) and
(\ref{3.2}) transformed Poisson structure $\eta:T^{\ast}(\mathcal{G}_{+}%
^{\ast}\times\mathcal{G}_{+}^{\ast})\longrightarrow T(\mathcal{G}_{+}^{\ast
}\times\mathcal{G}_{+}^{\ast})$ at\textrm{\ }$(F,Q)\in\mathcal{G}_{+}^{\ast
}\times\mathcal{G}_{+}^{\ast},$ \ where
\begin{align}
\eta &  =T\mathcal{\vartheta}T^{\ast},\nonumber\\
T  &  =\Phi_{(\tilde{l},l)}^{\prime}\Phi_{(Q,F)}^{{\Large \prime}-1}.
\label{3.3}%
\end{align}
Making use of the expressions
\[
\Phi_{(Q,F)}^{\prime}=\binom{-\,(.)F^{-1}\quad\tilde{l}(.)F^{-1}}%
{-F^{-1}(.)\quad F^{-1}(.)l},\quad\Phi_{(l,\widetilde{l})}^{{\LARGE \prime}%
}=\binom{1\quad0}{0\quad1},
\]%
\[
\Phi_{(Q,F)}^{{\Large \prime}-1}=\binom{-\,(1-\tilde{l}\otimes l^{-1}%
)^{-1}(.)F\quad\quad(1-\tilde{l}\otimes l^{-1})^{-1}\tilde{l}F(.)l^{-1}%
}{-\,(1-\tilde{l}\otimes l^{-1})^{-1}(.)F\quad\quad(1-\tilde{l}\otimes
l^{-1})^{-1}F(.)},
\]%
\begin{equation}
(\Phi_{(Q,F)}^{{\Large \prime\ast}})^{-1}=\binom{-\,F(.)(1-\tilde{l}%
^{-1}\otimes\tilde{l})^{-1}\quad\quad-\,F(.)(1-l^{-1}\otimes\tilde{l})^{-1}%
}{l^{-1}(.)F\tilde{l}(1-l^{-1}\otimes\tilde{l})^{-1}\quad\quad(.)F(1-l^{-1}%
\otimes\tilde{l})^{-1}} \label{3.4}%
\end{equation}
jointly with the $\vartheta$- structure (\ref{3.1}), one gets from (\ref{3.3})
that
\begin{align}
\eta &  =\binom{-\,(1-\tilde{l}\otimes l^{-1})^{-1}(.)F\quad(1-\tilde
{l}\otimes l^{-1})^{-1}\tilde{l}F(.)l^{-1}}{-\,(1-\tilde{l}\otimes
l^{-1})^{-1}(.)F\quad(1-\tilde{l}\otimes l^{-1})^{-1}F(.)}\times\nonumber\\
&  \times\left(
\begin{array}
[c]{c}%
\left[  \tilde{l},\left(  (1-\tilde{l}\otimes l^{-1})^{-1}(.)F(1-\tilde
{l}\otimes l^{-1})^{-1}(.)\right)  _{+}\right]  -\\
\left[  \left(  \tilde{l}^{-1}(.)F\tilde{l}(1-l^{-1}\otimes\tilde{l}%
)^{-1}\right)  _{+},l\right]  -\left[  l^{-1}(.)F\tilde{l}(1-l^{-1}%
\otimes\tilde{l})^{-1},l\right]  -.
\end{array}
\right. \nonumber\\
&  \left.
\begin{array}
[c]{c}%
-\left[  \tilde{l},(1-\tilde{l}\otimes l^{-1})^{-1}(.)F(1-\tilde{l}\otimes
l^{-1})^{-1}(.)\right]  _{+},\\
-\left[  l^{-1}(.)F\tilde{l}(1-l^{-1}\otimes\tilde{l})^{-1},l\right]  +\left[
\left(  (.)F(1-l^{-1}\otimes\tilde{l})^{-1}\right)  _{+},l\right]  ,
\end{array}
\right. \label{3.5}\\
&  \left.
\begin{array}
[c]{c}%
-\left[  \left(  F(.)(1-l^{-1}\otimes\tilde{l})^{-1}\right)  _{+},\tilde
{l}\right]  +\left[  \left(  F(.)(1-l^{-1}\otimes\tilde{l})^{-1}\right)
,\tilde{l}\right]  _{+}\\
\left[  \left(  (.)F(1-l^{-1}\otimes\tilde{l})^{-1}\right)  _{+},l\right]
-\left[  (.)F(1-l^{-1}\otimes\tilde{l})^{-1},l\right]  _{+}%
\end{array}
\right) \nonumber
\end{align}
at $\tilde{l}=QF^{-1}$ and $l=F^{-1}Q\in\mathcal{G}^{\ast}.$

Let now take any Casimir functional $\gamma\in\mathrm{I}(\mathcal{G}^{\ast}).$
Then one construct from the Poisson bracket (\ref{3.5}) the following
Hamiltonian flow on $\mathcal{G}_{+}^{\ast}\times\mathcal{G}_{+}^{\ast}$ $:$
\begin{equation}
\frac{d}{dt}(Q,F)^{\intercal}=\eta\nabla\gamma(Q,F), \label{3.6}%
\end{equation}
where $(Q,F)\in\mathcal{G}_{+}^{\ast}\times\mathcal{G}_{+}^{\ast}$ and
$t\in\mathbb{R}$ is the temporal evolution parameter. The flow (\ref{3.6} ) is
characterised by the following theorem.

\begin{theorem}
The Hamiltonian vector field $d/dt$ on $\mathcal{G}_{+}^{\ast}\times
\mathcal{G}_{+}^{\ast},$ defined by (\ref{3.6}), and the vector field $d/dt$
on $\mathcal{G}_{+}^{\ast}\times\mathcal{G}_{+}^{\ast},$ defined by
(\ref{d6.7}), coincide.
\end{theorem}

\begin{proof}
Proof of this theorem consists in simple but a slightly tedious calculation of
the expression \ (\ref{3.6}).
\end{proof}

\begin{remark}
The theorem above solves completely a problem posed in \cite{Dick} about
Hamiltonian formulation of the factorized differential-operator equations
(\ref{d6.7}).
\end{remark}

\subsection{Example 1}

We consider the following pseudo-differential factorized expression
\begin{equation}
l(\partial)=(\partial+u)^{-1}[(\partial+u)(\partial^{2}+2v)-2w]\ \label{d7.10}%
\end{equation}
for $F_{1}:=\partial+u,Q_{3}:=(\partial+u)(\partial^{2}+2v)-2w\in
\mathcal{G}_{+},(u,v,w)\in C^{\infty}(\mathbb{S}^{1};\mathbb{R}^{3}).$ The
respectively factorized differential operator evolution equations
\ (\ref{d7.7}) give rise to the following \cite{BoLiXi} interesting system
\begin{align}
u_{t}  &  =2uu_{x}+2v_{x}-u_{xx},\label{d7.11}\\
v_{t}  &  =2w_{x},\text{ \ \ }w_{t}=w_{xx}+2(wu)_{x}\nonumber
\end{align}
of completely integrable evolution equations.

\begin{remark}
It is worth to mention that the derived above system \ of integrable equations
(\ref{d7.11}) allows the following degenerate purely differential-matrix
linear spectral problem:%
\begin{equation}
\left(
\begin{array}
[c]{cc}%
(\partial^{2}+2v)(\partial-u) & -\lambda\\
\partial-u & -1
\end{array}
\right)  \left(
\begin{array}
[c]{c}%
f\\
g
\end{array}
\right)  =0\ \label{d7.12}%
\end{equation}
for $(f,g)^{\intercal}\in L_{2}(\mathbb{S}^{1};\mathbb{C}^{2})$ and arbitrary
spectral parameter $\lambda\in\mathbb{C}.$
\end{remark}

\subsection{Example 2}

A next example is related with the pseudo-differential factorized expression
\begin{equation}
l(\partial)=[(\partial+w)(\partial+p)]^{-1}[(\partial+w)(\partial
+p)\partial+(\partial+p)u+v]\ \label{d7.13}%
\end{equation}
for $F_{2}:=(\partial+w)(\partial+p)$ and $Q_{3}:=(\partial+w)(\partial
+p)\partial+(\partial+p)u+v\in\mathcal{G}_{+},(u,p,v,w)\in C^{\infty
}(\mathbb{S}^{1};\mathbb{R}^{4}).$ From the factorized differential operator
evolution equations \ (\ref{d7.7}) one easily \ ensues the system
\begin{align}
u_{t}  &  =u_{2x}+2v_{x}+2(uw)_{x},\label{d7.14}\\
v_{t}  &  =v_{2x}+2vw_{x}+2(pw)_{x},\text{ \ \ }\nonumber\\
w_{t}  &  =-w_{xx}+2u_{x}+2ww_{x},\nonumber\\
p_{t}  &  =-p_{xx}-2w_{2x}+2u_{x}+2pp_{x},\nonumber
\end{align}
of completely integrable evolution flows on $C^{\infty}(\mathbb{S}%
^{1};\mathbb{R}^{4}),$ considered also before in \cite{SzBl,BoLiXi} in the
context of generating a new class of integrable dispersionless systems of
hydrodynamic type equations.

\subsection{Example 3}

\begin{quotation}
Let us put now the following pseudo-differential factorized expression%
\begin{equation}
l(\partial)=\partial+(1/4-\alpha^{2}\partial^{2})^{-1}(\gamma\partial
^{2}+v/2\ +\beta/4)+\gamma\alpha^{-2}, \label{d7.15}%
\end{equation}
where $\alpha,\beta$ and $\gamma\in\mathbb{R}$ are constants, $v\in C^{\infty
}(\mathbb{S}^{1};\mathbb{R}),F_{2}:=1/4-\alpha^{2}\partial^{2}$ and
$Q_{3}:=\gamma\partial^{2}+v/2\ +\beta/4\in\mathcal{G}_{+}.$ The related
factorized differential operator evolution equations\ (\ref{d7.7}) are reduced
for the gradient $\ $element $\nabla\gamma(l)_{+}=\nabla\gamma(l)-$
$\nabla\gamma(l)_{-}=\ \frac{1}{2}u_{x}\ -u\partial\ -$ $(1/4-\alpha
^{2}\partial^{2})^{-1}(\gamma\partial^{3}+\partial\ v/2\ +\beta\partial
/4)\in\mathcal{G}_{-}\oplus\{\partial\},$ where and $u:=(1-\alpha^{2}%
\partial^{2})^{-1}v\in C^{\infty}(\mathbb{S}^{1};\mathbb{R})$ and the element
$\partial\in\mathcal{G}$ \ is a character of the Lie algebra $\mathcal{G},$
that is $(\partial,[\mathcal{G}_{\pm},\mathcal{G}_{\pm}])=0,\ $\ to the
following evolution flow:%
\begin{equation}
v_{t}+\beta u_{x}+uv_{x}+2vu_{x}+\gamma u_{3x}=0,\ \label{d7.16}%
\end{equation}
describing simple wave motion \cite{CaHo} of the Euler equations for shallow
water dynamics.
\end{quotation}

\textbf{Acknowledgements}

The authors cordially thank Prof. M. B\l aszak, Prof. B. Szablikowski and
Prof. A. Samoilenko for the cooperation and useful discussion of the results
in this paper during the International Conference in Functional Analysis
dedicated to the 125$^{th}$ anniversary of Stefan Banach held on 18 - 23
September, 2017 in Lviv, Ukraine. Especially authors are thankful to Prof. M.
Pavlov for the interest in our work, instrumental remarks and presenting very
important references.

\end{document}